%% file: main.tex
\documentclass[11pt,letterpaper]{article}

\usepackage{amsmath,amsfonts,amssymb,color}
\usepackage{amsthm}
\usepackage[margin=1in]{geometry}
\usepackage{graphics} 
\usepackage{tikz,,subcaption,cite}
\usetikzlibrary{positioning}
\usetikzlibrary{arrows}
\usepackage{natbib} 
\usepackage{epstopdf}
\usepackage{epsfig} 
\usepackage{psfrag}
\usepackage[width=.85\textwidth]{caption}
\usepackage{lipsum}
\usepackage{algorithm} 
\usepackage{algorithmic}
\usepackage[title,titletoc,toc]{appendix}
\usepackage{url}
\usepackage[width=.85\textwidth, font=footnotesize]{caption}
\usepackage[pdfpagelabels,pdfpagemode=None,breaklinks=true]{hyperref} \hypersetup{
     colorlinks   = true,
     citecolor    = blue
}
\renewcommand{\cite}{\citep}

\newtheorem{theorem}{Theorem}
\newtheorem{lemma}{Lemma}

\newtheorem{definition}{Definition}

\makeatletter
\renewcommand\bibsection%
{
  \section*{\refname
    \@mkboth{\MakeUppercase{\refname}}{\MakeUppercase{\refname}}}
}
\makeatother

\newcommand{\ignore}[1]{}

\DeclareMathAlphabet{\mathcal}{OMS}{cmsy}{m}{n}      
      
\title{\LARGE \bf Ignoring Extreme Opinions in Complex Networks: \\  The Impact of Heterogeneous Thresholds}
\author{Shreyas~Sundaram%
\thanks{The author is with the School of Electrical and Computer Engineering at Purdue University. E-mail: {\tt sundara2@purdue.edu}.}%
}

\begin{document}

\date{}
\maketitle
\thispagestyle{empty}
\pagestyle{empty}

\begin{abstract}
We consider a class of opinion dynamics on networks where at each time-step, each node in the network disregards the opinions of a certain number of its most extreme neighbors  and updates its own opinion as a weighted average of the remaining opinions.  When all nodes disregard the same number of extreme neighbors, previous work has shown that consensus will be reached if and only if the network satisfies certain topological properties.  In this paper, we consider the implications of allowing each node to have a personal threshold for the number of extreme neighbors to ignore.  We provide graph conditions under which consensus is guaranteed for such dynamics.  We then study random networks where each node's threshold is drawn from a certain distribution, and provide conditions on that distribution, together with conditions on the edge formation probability, that guarantee that consensus will be reached asymptotically almost surely.  
\end{abstract}

%%%%%%%%%%%%%%%%%%%%%%%%%%%%%%%%%%%%%%%%%%%%%%%%%%%%%%%%%%%%%%
\section{Introduction}
\label{sec:intro}

\input{intro}

\subsection*{Notation and Definitions}

A graph (or network) is denoted by $\mathcal{G}=\{\mathcal{V},\mathcal{E}\}$, where $\mathcal{V}=\{v_1, \ldots, v_n\}$ is a set of nodes (or vertices or individuals)  and $\mathcal{E} \subseteq \mathcal{V} \times \mathcal{V}$ is a set of edges.  An edge $(v_j, v_i)$ indicates that node $v_i$ can receive information from node $v_j$.  The set of {\it neighbors} of node $v_i$ is defined as $\mathcal{N}_i=\{v_j\in\mathcal{V}  \mid (v_j, v_i) \in \mathcal{E}\}$.  Correspondingly, the {\it degree} of node $v_i$ is given  by $d_i =| \mathcal{N}_i|$, and the {\it minimum degree} of the network is $\min_{v_i\in\mathcal{V}}d_i$.  A {\it path} from node $v_j$ to $v_i$ is a sequence of nodes $v_{i_1}, v_{i_2}, \ldots, v_{i_k}$ such that $v_{i_1} = v_j$, $v_{i_k} = v_i$ and $(v_{i_l},v_{i_{l+1}}) \in \mathcal{E}$ for $l = 1, 2, \ldots, k-1$.  The graph is said to be {\it strongly connected} if there is a path from every node to every other node.  The {\it connectivity} of a graph is the smallest number of nodes that have to be removed in order to cause the remaining network to not be strongly connected.   

In our derivations, we use $\mathbb{Z}$ to denote the set of integers, $\mathbb{N}$ to denote the set of nonnegative integers, and $\mathbb{R}$ to denote the set of real numbers.  We add subscripts to these quantities to denote restrictions of the sets to appropriate values.  For two functions $f: \mathbb{N} \rightarrow \mathbb{R}$ and $g: \mathbb{N} \rightarrow \mathbb{R}$, we say  $f(n)=O(g(n))$ if there exist constants $c>0$ and $n_0$ such that $|f(n)| \le c|g(n)|$ for all $n \ge n_0$.  We say $f(n)=o(g(n))$ if  $\frac{f(n)}{g(n)}\to 0$ as $n\to\infty$.

%%%%%%%%%%%%%%%%%%%%%%%%%%%%%%%%%%%%%%%%%%%%%%%%%%%%%%%%%%%%%%%%%%%%%%%%%%%%%%%
\section{Filtering-Based Opinion Dynamics}
Consider a population of individuals modeled by the directed network $\mathcal{G} = \{\mathcal{V},\mathcal{E}\}$.  Each individual (node) $v_i$ has an initial opinion $x_i[0] \in \mathbb{R}$.  We assume that time progresses in a sequence of quantized increments, referred to as {\it time-steps}.  At each time-step $k \in \mathbb{N}$, each node receives the opinions of its neighbors, and updates its opinion as a function of those received values.  

In this paper, we define a {\it filtering threshold} $t_i \in \{0, 1, \ldots, n-1\}$ for each node $v_i \in \mathcal{V}$.  Based on this threshold, we consider the following class of opinion dynamics.
\begin{enumerate}
\item At each time-step $k$, each node $v_i$ gathers the opinions of its neighbors, i.e., $ \{x_j[k] \mid v_j \in \mathcal{N}_i\}$.
\item Each node $v_i$ removes the $t_i$ largest opinions in its neighborhood that are higher than its own opinion (if there are fewer than $t_i$ such opinions, node $v_i$ only disregards those opinions).  Each node $v_i$ also removes the $t_i$ smallest opinions in its neighborhood that are smaller than its own opinion (if there are fewer than $t_i$ such opinions, node $v_i$ only disregards those opinions).  Ties in opinions are broken arbitrarily.  Let $\mathcal{M}_i[k] \subset \mathcal{N}_i$ denote the set of {\it moderate} neighbors of $v_i$ at time-step $k$ (i.e., those nodes whose opinions were not discarded).
\item Each node $v_i$ updates its own opinion as
\begin{equation}
x_i[k+1] = w_{ii}[k]x_i[k] + \sum_{v_j \in \mathcal{M}_{i}[k]}w_{ij}[k]x_j[k] ,
\label{eq:LF_dynamics}
\end{equation}
where the weights $w_{ii}[k]$ and $w_{ij}[k]$, $v_j \in \mathcal{M}_{i}[k]$ are nonnegative, lower bounded by some positive constant $\eta$, and sum to $1$.  
\end{enumerate}

Note that the set $\mathcal{M}_{i}[k]$ can be empty if node $v_i$ has fewer than $2t_i$ neighbors.  In particular, if $t_i = d_i$ for some $v_i$, that node disregards all of its neighbors and thus becomes a {\it stubborn node}.  When $t_i = 0$ for all $v_i \in  \mathcal{V}$, the above dynamics reduce to classical DeGroot dynamics.  In this case, all nodes will reach consensus as long as the graph is strongly connected (or more generally, contains a spanning tree rooted at some node) \cite{DeGroot74, Ren05}.  However, when the nodes have nonzero filtering thresholds, simply having a strongly connected graph is no longer sufficient to ensure consensus.  To illustrate, consider the network shown in Fig.~\ref{fig:Counterexample} where every node has filtering threshold $1$.  This network consists of two complete subgraphs on node sets $\mathcal{S}_1$ and $\mathcal{S}_2$.   Each node in $\mathcal{S}_1$ has exactly one neighbor in the opposite set.  This network is strongly connected; in fact, it has connectivity $\frac{n}{2}$ and minimum degree $\frac{n}{2}$.  However suppose that each node in $\mathcal{S}_1$ has initial opinion $0$ and each node in $\mathcal{S}_2$ has initial opinion $1$.  Under the filtering dynamics described above, each node in each community removes the opinion of its neighbor from the opposite community, and thus no node ever changes its opinion.  Thus, consensus is not reached in this network even when each node disregards only a single highest and lowest opinion in its neighborhood at each time-step. 

\input{counterexample_graph}

In \cite{leblanc13}, we considered the above class of filtering dynamics for the case where the thresholds for each node are the same (i.e., $t_i = t_j$ for all $v_i, v_j \in \mathcal{V}$), but where certain nodes in the network were allowed to be {\it malicious} and deviate from the dynamics in arbitrary ways.  We established graph-theoretic properties that will ensure consensus of the non-malicious nodes under these conditions.  In the next section, we will generalize these conditions to the case of heterogeneous thresholds, and establish necessary and sufficient graph conditions for consensus under the filtering dynamics \eqref{eq:LF_dynamics}.  Subsequently, we will use these graph conditions to study this class of opinion dynamics in random graph models for complex networks, with randomly chosen filtering thresholds for each node.

%%%%%%%%%%%%%%%%%%%%%%%%%%%%%%%%%%%%%%%%%%%%%%%%%%%%%%%%%%%%%%%%%%
\section{Graph-Theoretic Conditions to Ensure Consensus Under Filtering-based Opinion Dynamics with Heterogeneous Thresholds}

Given a network $\mathcal{G} = \{\mathcal{V},\mathcal{E}\}$ with $n$ nodes, let $\mathcal{T} = \{t_1, t_2, \ldots, t_n\}$ be the set of filtering thresholds.  We define the following notions.

\begin{definition}
We say that a set $\mathcal{S} \subset\mathcal{V}$ is {\bf $(\mathcal{T}+1)$-reachable} if there exists a node $v_i \in \mathcal{S}$ that has at least $t_i+1$ neighbors outside $\mathcal{S}$, i.e., $|\mathcal{N}_i \setminus \mathcal{S}| \ge t_i+1$.
\end{definition}

\begin{definition}
We say that network $\mathcal{G}$ is {\bf $(\mathcal{T}+1)$-robust} if for every pair of disjoint nonempty subsets $\mathcal{S}_1, \mathcal{S}_2 \subset \mathcal{V}$, either $\mathcal{S}_1$ or $\mathcal{S}_2$ is $(\mathcal{T}+1)$-reachable.
\end{definition}

Note that when $t_i = 0$ for all $v_i \in \mathcal{V}$, then the network being $(\mathcal{T}+1)$-robust is equivalent to it having a spanning tree rooted at some node.  The above definitions are a relatively straightforward extension of the notions of robust networks given in \cite{leblanc13,zhang2015notion}, where each node had the same threshold.  The intuition behind the definition of a $(\mathcal{T}+1)$-robust graph follows from examining the failure of consensus in Fig.~\ref{fig:Counterexample}:  that network contained two subsets of nodes where each node in each subset filtered away the only information it received from the opposite subset.  In order to ensure consensus, we would like to avoid such situations.  This is captured by the following theorem.

\begin{theorem}
Consider a network $\mathcal{G} = \{\mathcal{V},\mathcal{E}\}$ with a set of filtering thresholds $\mathcal{T}$.  Then consensus is guaranteed under the filtering-based opinion dynamics \eqref{eq:LF_dynamics} regardless of the initial opinions if and only if the network is $(\mathcal{T}+1)$-robust.
\label{thm:LF_dynamics_consensus}
\end{theorem}

\begin{proof}
For the proof of necessity, suppose the network is not $(\mathcal{T}+1)$-robust.  Then there exist two disjoint nonempty subsets $\mathcal{S}_1$ and $\mathcal{S}_2$ of nodes such that neither set is $(\mathcal{T}+1)$-reachable.  Let the initial opinions of the nodes in set $\mathcal{S}_1$ be $0$, and let the initial opinions of the nodes in set $\mathcal{S}_2$ be $1$.  Let the initial opinions of all nodes in set $\mathcal{V}\setminus \{\mathcal{S}_1\cup\mathcal{S}_2\}$ be $0.5$.  Now, since each node $v_i  \in \mathcal{S}_1$ has at most $t_i$ neighbors outside $\mathcal{S}_1$, each node will remove all of those opinions when it updates its opinion, and thus its opinion will stay at $0$.  The same reasoning holds for the set $\mathcal{S}_2$, and thus consensus will never be reached.

The proof of sufficiency follows from a straightforward generalization of the proof of sufficiency under homogeneous thresholds given in \cite{leblanc13}, and thus we omit it here.
\end{proof}

\iffalse only give a sketch of the proof here.  At each time-step $k$, consider the sets $\mathcal{V}_M[k]$ and $\mathcal{V}_m[k]$ to be the nodes that have the maximum and minimum opinions, respectively.  These are disjoint nonempty sets, and since the network is $\mathcal{T}$-robust, at least one of these sets will have a node $v_i$ that has $t_i+1$ neighbors outside its set.  Since node $v_i$ only disregards $t_i$ nodes that are larger (or smaller) than its own value, $v_i$ will use the opinion of at least one node that is different from its own opinion, and thus will get pulled out of the set of maximum or minimum values.  By continuing this process, one can show that the maximum and minimum opinions in the network asymptotically converge towards each other, thereby ensuring consensus.\fi

Armed with the above characterization of conditions under which the opinion dynamics \eqref{eq:LF_dynamics} lead to consensus, we now study such dynamics in random graph models for large-scale networks.

%%%%%%%%%%%%%%%%%%%%%%%%%%%%%%%%%%%%%%%%%%%%%%%%%%%%%%%%%%%%%%%%%%%%%%%%%%%%%%%
\section{Opinion Dynamics in Erd\H os-R\'enyi Random Graphs}
\label{sec:E-R}
\input{ER}

%%%%%%%%%%%%%%%%%%%%%%%%%%%%%%%%%%%%
\section{Opinion Dynamics in Random Graphs with Arbitrary Community Structure}
\label{sec:rin}
\input{interdep}

%%%%%%%%%%%%%%%%%%%%%%%%%%%%%%%%%%%

\section{Opinion Dynamics in Random Graphs with Heterogeneous Degree Distributions}
\label{sec:hetero}
\input{hetero}

%%%%%%%%%%%%%%%%%%%%%%%%%%%%%%%%%%

%%%%%%%%%%%%%%%%%%%%%%%%%%%%%%%%%%

\section{Summary and Future Work}
We studied a class of opinion dynamics where each node ignores the most extreme opinions in its neighborhood at each time-step.  We allowed each node to have a personal threshold for the number of neighbors that it ignores, and provided necessary and sufficient conditions on the network topology that guarantee consensus under such dynamics.  We then studied random graph models where each node has a random threshold that is drawn from a certain distribution.  We characterized properties of that distribution (in terms of the edge probabilities of the underlying network) that led to the network satisfying the required conditions for consensus.  Our analysis encompassed classical Erd\H os-R\'enyi networks, as well as networks with arbitrary community structure and networks with heterogeneous edge probabilities.  

There are a variety of interesting directions for future research, including a study of other classes of random graphs (and threshold distributions), along with tightness characterizations of the conditions that we have provided on the threshold distribution. 

%%%%%%%%%%%%%%%%%%%%%%%%%%%%%%%

\sloppy
\bibliographystyle{plainnat}
{\footnotesize \bibliography{refs}}

\end{document}

%% file: intro.tex
The study of how opinions, fads, and ideas spread through populations has received significant attention over the past several decades, spanning diverse disciplines including sociology, mathematics, physics, computer science and engineering.  This body of work has shown that complex global phenomena can arise as a result of relatively simple interaction rules for the individuals in the population.  In particular, there is a tight coupling between the {\it structure} of the interactions (i.e., who talks to whom) and the {\it dynamics} of these interactions (i.e., what happens when they talk).

An early investigation of {\it opinion dynamics} in networks was initiated by DeGroot \cite{DeGroot74}, who proposed that individuals repeatedly update their personal (real-valued) opinions as a weighted average of their neighbors' opinions.  Such dynamics lead to consensus under mild assumptions on the underlying network topology, and there has been substantial effort devoted to extending such averaging rules to settings involving time-varying graph topologies, higher-order dynamics, and stubborn or malicious individuals \cite{Jadbabaie03, Moreau05, Ren05, SundaramHadjicostis11,acemoglu2013opinion,friedkin2015problem,pirani2016smallest}.  A modification of the classical DeGroot model is the {\it bounded confidence} model of Hegselmann and Krause (and separately, Weisbuch and Deffuant), where an individual only averages the opinions of those neighbors whose opinions are close (according to some metric) to their own \cite{HK2002,weisbuch2002meet,amblard2004role,etesami2013termination,Ts2009}.   In such cases, the opinions in the network settle down to a stratified set of values (i.e., not necessarily in consensus), depending on the initial spread of opinions.  

In addition to the above studies of real-valued opinions, researchers have also investigated the diffusion of {\it binary} valued information;  this models, for instance, the adoption of a certain innovation, or deciding whether to participate in an activity (such as a riot) \cite{Granovetter1978,schelling2006micromotives, Morris00, watts2002simple,sood2005voter,masuda2015opinion,yildiz2013binary}.  The dynamics in such cases are often manifested as {\it threshold-based} rules, where an individual adopts the action if a certain number or fraction of their neighbors have done so.  In particular, {\it heterogeneous thresholds} play an important role in such dynamics, as they lead to cascades whereby individuals with low thresholds initiate the adoption process, and individuals with increasingly higher thresholds subsequently join the cascade \cite{Granovetter1978, watts2002simple, gladwellTipping, Goldenberg01}.  

In this paper, we study a generalization of DeGroot opinion dynamics by introducing a threshold-based filtering rule into the averaging dynamics.  Specifically, instead of each node averaging {\it all} of its neighbors opinions (as in DeGroot), we consider the scenario where each individual disregards the most extreme opinions in its neighborhood and averages only the opinions of its moderate neighbors.  This rule differs from the bounded confidence models of Hegselmann and Krause in that each individual discards only a certain number of most extreme values in its neighborhood under our dynamics, as opposed to any number of values that differ significantly from its own opinion.  This type of filtering rule has been recently studied in the context of resilient consensus dynamics (where some nodes can be faulty or malicious), and all non-adversarial nodes use the {\it same} threshold (i.e., they all disregard the same number of extreme neighbors) \cite{leblanc13,Nitin12}.  These previous works have established graph-theoretic conditions that are required for such rules to guarantee consensus among the non-adversarial nodes, regardless of the actions of the adversaries.  In this paper we do not consider adversarial behavior, but instead study the impact of {\it heterogeneity} in the filtering thresholds (motivated by the study of such heterogeneous thresholds in binary diffusion dynamics, as highlighted above).  Our contributions are as follows.  We first extend the graph-theoretic characterizations from \cite{leblanc13} to the case of heterogeneous thresholds and provide necessary and sufficient conditions on the network topology (in terms of the thresholds) for consensus to be guaranteed regardless of the initial opinions.  We then study such dynamics in random graph models for complex networks, where each node's opinion filtering threshold is chosen independently and randomly from a certain distribution.  We characterize properties of this distribution such that with high probability, the resulting network guarantees consensus;  we study Erd\H os-R\'enyi random graphs (where each edge is placed independently with the same probability), random interdependent networks (consisting of various subnetworks or communities with arbitrary topologies), and random graphs with heterogeneous edge probabilities.

%% file: counterexample_graph.tex
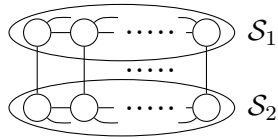
\begin{figure}[b]
\begin{center}
\begin{tikzpicture}[scale = 0.25]

  \node [circle, draw] (n1) at (1,1)  {};
  \node [circle, draw] (n2) at (3.5,1)  {};
  \node [circle, draw] (n3) at (10,1)  {};
  \node [circle, draw] (n4) at (1,5)  {};
  \node [circle, draw] (n5) at (3.5,5)  {};
  \node [circle, draw] (n6) at (10,5)  {};
  \foreach \from/\to in {n1/n2,n1/n4,n2/n5,n3/n6,n4/n5}
      \draw (\from) -- (\to);

  \node (n7) at (7,1) {{\Large $.....$}};
  \draw (n2) -- (n7);
  \draw (n7) -- (n3);
  \node (n8) at (7,5) {{\Large $.....$}};
  \draw (n5) -- (n8);
  \draw (n8) -- (n6);
  \node (n9) at (7,3) {{\Large $.....$}};
  \draw (5.5,1) ellipse (6cm and 1.5cm);
  \draw (5.5,5) ellipse (6cm and 1.5cm);
  \node at (13,5) {$\mathcal{S}_1$};
  \node at (13,1) {$\mathcal{S}_2$};

  \draw    (n1) to[out=-45,in=180] (2.8,0.2);
  \draw    (n2) to[out=-45,in=180] (5.3,0.2);
  \draw    (n4) to[out=45,in=180] (2.8,5.8);
  \draw    (n5) to[out=45,in=180] (5.3,5.8);
  \draw    (n3) to[out=-135,in=0] (8.5,0.2);
  \draw    (n6) to[out=135,in=0] (8.5,5.8);

\end{tikzpicture}
\caption{Example of a network where consensus is not guaranteed when each node ignores the single highest and single lowest opinion in its neighborhood at each time-step.}
\label{fig:Counterexample}
\end{center}
\end{figure}

%% file: ER.tex
We start by considering the outcome of the opinion dynamics \eqref{eq:LF_dynamics} in Erd\H os-R\'enyi random graphs, defined as follows.

\begin{definition}
For $n \in \mathbb{Z}_{\ge 1}$, let $\Omega_n$ be the set of all undirected graphs on $n$ nodes, and define $p(n) \in [0,1]$.  Define the probability space $(\Omega_n, \mathcal{F}_n, \mathbb{P}_n)$, where the $\sigma$-algebra $\mathcal{F}_n$ is the power set of $\Omega_n$, and $\mathbb{P}_n$ is a probability measure that assigns the probability
$$
\mathbb{P}_n(\mathcal{G}) = p(n)^{m}(1-p(n))^{{n\choose{2}}-m}
$$
to each graph $\mathcal{G}$ with $m$ edges.  Then, a graph drawn from $\Omega_n$ according to the above probability distribution is called an {\bf Erd\H os-R\'enyi} (ER) random graph, and denoted $\mathcal{G}_{n,p}$.  Equivalently, an ER graph is obtained by placing each edge in the graph independently with probability $p(n)$.
\end{definition}

\begin{definition}
Let $R$ be a graph property, and let $G_n^R \subseteq \Omega_n$ be the set of graphs on $n$ nodes that have that property. We say an ER graph has property $R$ {\bf asymptotically almost surely} (a.a.s.) if  $\lim_{n\rightarrow\infty}\mathbb{P}_n(G_n^R) = 1$. 
\end{definition}

An important feature of  $\mathcal{G}_{n,p}$ is that the model displays  `phase transitions' for certain properties. Loosely speaking, if the probability of adding an edge is `larger' than a certain value (which could be a function of $n$) then the ER graph  will have a certain property a.a.s., and if the edge probability is `smaller' than that value, then the graph will a.a.s. not have that property. We make this more precise for the following properties.

\begin{definition}
For any $r\in \mathbb{Z}_{\geq1}$, let $\mathcal{D}_r$ be the property of having minimum-degree $r$ and let $\mathcal{K}_r$ be the property of having connectivity $r$.  \end{definition}

\begin{lemma}[\cite{ErdosRenyi1961}]
For any constant $r\in \mathbb{Z}_{\geq1}$, let the edge probability be given by 
$$
p(n)= \frac{\ln n + (r-1)\ln\ln n + c(n)}{n},
$$ 
where $c(n)$ is some function of $n$.  If $c(n) \rightarrow \infty$ as $n \rightarrow \infty$, the ER graph has properties $\mathcal{D}_r$ and  $\mathcal{K}_r$  a.a.s.  If $c(n) \rightarrow -\infty$ as $n \rightarrow \infty$, then the ER graph a.a.s. does not have either of the properties $\mathcal{D}_r$ or $\mathcal{K}_r$.
\label{lem:threshold_connectivity}
\end{lemma}

Based on the above result, we see that the ER graph requires $p(n) = \frac{\ln n + c(n)}{n}$ with $c(n) \rightarrow \infty$ in order to be connected a.a.s. (as this corresponds to the property $\mathcal{K}_1$).  This is the regime that we will focus on here, as consensus under generic initial opinions cannot be obtained in disconnected networks, regardless of the filtering thresholds at the nodes.   In particular, in our proofs, we will take $r \in \mathbb{Z}_{\ge 1}$ to be the largest integer such  that $p(n) \ge \frac{\ln n + (r-1)\ln\ln n + c(n)}{n}$ where $c(n) \rightarrow \infty$; this $r$ is then the minimum degree of the graph.   If this is true for all $r \in \mathbb{N}$, then our analysis holds for any positive integer $r$.  

As we are interested in the effects of {\it heterogeneous filtering thresholds}, we will consider the case where the threshold for each node is drawn independently from a given distribution $q(\cdot)$ with support $[0, \bar{r}]$, where $\bar{r} \in \mathbb{Z}_{\ge 1}$.  Specifically, the probability that node $v_i$ has filtering threshold $t$ is given by $q(t)$.  There are a few subtle details underlying this analysis.  First, note that if a node has a filtering threshold that is larger than its degree, then consensus will not be guaranteed in general (e.g., if that node has the largest opinion in the network, it will disregard all of its neighbors, and all other nodes with nonzero thresholds will disregard the extreme opinion of that node).  On the other hand, even if the filtering thresholds are much smaller than the node degrees, the graph may not be $(\mathcal{T}+1)$-robust; for example, in the network shown in Fig.~\ref{fig:Counterexample}, all nodes have degree $\frac{n}{2}$ and filtering threshold $1$, but the graph is not $(\mathcal{T}+1)$-robust.  Thus, the filtering threshold distribution must be such that these events occur with vanishing probability.  We provide such a characterization below.  

\begin{theorem}
Consider an ER graph with edge probability 
\begin{equation}
p(n) = \frac{\ln n + (r-1)\ln\ln n + c(n)}{n},
\label{eq:ER_edge_prob}
\end{equation}
where $c(n) = o(\ln\ln n)$ and goes to $\infty$ with $n$.  Suppose the filtering threshold distribution has support $[0, \bar{r}]$ for some $\bar{r} \in \mathbb{N}$, and satisfies
\begin{equation}
q(t) = O\left(\frac{1}{{(\ln n)}^{t-r+1}}\right)
\label{eq:ER_filtering_threshold}
\end{equation}
for $t \ge r$.  Then the ER random graph with filtering thresholds drawn from $q(\cdot)$ will facilitate consensus under the local-filtering opinion dynamics a.a.s.
\label{thm:ER_graph_filtering_thresholds}
\end{theorem}

\begin{proof}
We will show that the ER graph is $(\mathcal{T}+1)$-robust a.a.s. under the conditions in the theorem,\footnote{The choice of $c(n) = o(\ln\ln n)$ in the theorem is for technical reasons, but is not restrictive; as we argue later in Section~\ref{sec:hetero}, any probability larger than the one given in this theorem will also suffice for the result to hold.}  where $\mathcal{T}$ is the vector of filtering thresholds, each drawn independently from the distribution $q(\cdot)$.   To do this, we will show that every subset of vertices of cardinality up to $\alpha n$ is $(\mathcal{T}+1)$-reachable, where $\alpha \in (0, 1)$ is some function of $n$ that goes to $1$.  This will be sufficient to prove $(\mathcal{T}+1)$-robustness of the network, since if we take any two disjoint nonempty sets, at least one of them will have size at most $\frac{n}{2} \le \alpha n$, and thus will be $(\mathcal{T}+1)$-reachable.

To this end, consider any subset $\mathcal{S}\subset\mathcal{V}$ of nodes with cardinality $m \in \{1, 2, \ldots, \alpha n\}$.  Consider some node $v_i \in \mathcal{S}$.  The probability that $v_i$ has less than $t_i + 1$ neighbors outside $\mathcal{S}$ (where $t_i$ is randomly chosen from the distribution $q(\cdot)$) is given by
\begin{equation*}
\sum_{t = 0}^{\bar{r}}q(t)\sum_{j = 0}^{t}{{n-m}\choose{j}}p^j(1-p)^{n-m-j},
\end{equation*}
which is upper bounded by
\begin{equation}
\sum_{t = 0}^{\bar{r}}q(t)\sum_{j = 0}^{t}{{n}\choose{j}}p^j(1-p)^{n-m-j}.
\label{eq:single_node_ER}
\end{equation}
Thus, the probability $\mathbb{P}_{\mathcal{S}}$ that all nodes $v_i$ in set $\mathcal{S}$ have fewer than $t_i+1$ neighbors outside $\mathcal{S}$ is upper bounded by 
\begin{equation}
\mathbb{P}_{\mathcal{S}} \le  \left(\sum_{t = 0}^{\bar{r}}q(t)\sum_{j = 0}^{t}{{n}\choose{j}}p^j(1-p)^{n-m-j}\right)^{m}.  \label{eqn:P_S}
\end{equation}
Let $b_j \triangleq {{n}\choose{j}}p^j(1-p)^{n-m-j}$.  We have
$$
\frac{b_j}{b_{j-1}} = \frac{n-j+1}{j}\frac{p}{1-p} \ge \frac{n-t}{t}\frac{p}{1-p}
$$
for $j \le t$, which goes to $\infty$ for $t = o(np)$.  Thus, for any $\beta > 1$, we have $\frac{b_j}{b_{j-1}} > \beta$ for sufficiently large $n$.  This yields
\begin{equation*}
\sum_{j = 0}^{t}{{n}\choose{j}}p^j(1-p)^{n-m-j} \le \sum_{j = 0}^{t}\frac{1}{\beta^{t-j}}b_t \le b_t\sum_{j =0}^{\infty}\frac{1}{\beta^j} = C_1b_t
\end{equation*}
for some positive constant $C_1$. Substituting this back into \eqref{eqn:P_S}, we have
$$
\mathbb{P}_{\mathcal{S}} \le \left(C_1\sum_{t = 0}^{\bar{r}}q(t){{n}\choose{t}}p^t(1-p)^{n-m-t}\right)^{m}.
$$
Using the fact that ${n\choose{t}} \le n^t$ and $1-p \le e^{-p}$, we have
\begin{equation*}
\mathbb{P}_{\mathcal{S}} \le  \left(C_1\sum_{t = 0}^{\bar{r}}q(t)(np)^te^{-np}e^{p(m+t)}\right)^{m}.
\end{equation*}
Substituting the expression for $p$ from the theorem, we have
\begin{equation*}
\mathbb{P}_{\mathcal{S}} \le  \left(C_1\sum_{t = 0}^{\bar{r}}q(t)\frac{(\ln n + (r-1)\ln\ln n + c)^t}{n(\ln n)^{r-1}}e^{-c}e^{p(m+t)}\right)^{m}.
\end{equation*}
Using condition \eqref{eq:ER_filtering_threshold} on the filtering threshold distribution given in the theorem, we have
$$
q(t)(\ln n + (r-1)\ln\ln n + c)^t\frac{1}{(\ln n)^{r-1}} = O(1).
$$
Thus, we have
$$
\mathbb{P}_{\mathcal{S}} \le \left(C_2\sum_{t = 0}^{\bar{r}}\frac{1}{n}e^{-c}e^{p(m+t)}\right)^{m} \le \left(C_3\frac{1}{n}e^{-c}e^{pm}\right)^{m}
$$
for some positive constants $C_2$ and $C_3$.

Let $\mathbb{P}_m$ be the probability that there is {\it some} set $\mathcal{S}\subset\mathcal{V}$ of size $m$ such that all nodes $v_i \in \mathcal{S}$ have fewer than $t_i+1$ neighbors outside $\mathcal{S}$.  By the union bound, we have
\begin{equation*}
\mathbb{P}_{m} \le {n\choose{m}}\mathbb{P}_{\mathcal{S}} 
\le \left(C_3\frac{ne}{m}\frac{1}{n}e^{-c}e^{pm}\right)^{m} = \left(C_3e^{1-c}\frac{e^{pm}}{m}\right)^{m},
\end{equation*}
where we used the fact that ${n\choose{m}} \le \left(\frac{ne}{m}\right)^m$.  It was shown in \cite{zhang2015notion} that the function $\frac{e^{pm}}{m} \le \max\{e^p, \frac{1}{\alpha n}e^{\alpha np}\}$ for $1 \le m \le \alpha n$.  For the edge probability $p$ of the form given in the theorem, we see that 
\begin{equation*}
\frac{1}{\alpha n}e^{\alpha np} = \frac{1}{\alpha}e^{\alpha np - \ln n} = \frac{1}{\alpha}e^{-(1-\alpha)\ln n +\alpha(r-1)\ln\ln n +\alpha c(n)} = o(1)
\end{equation*}
whenever $\ln\ln n = o((1-\alpha)\ln n)$.  Let $\alpha$ be such that this holds.  Then we have $\frac{e^{pm}}{m} \le e^p \le e$ for sufficiently large $n$.  Thus, we have $\mathbb{P}_m \le \left(C_4e^{-c}\right)^m$ 
for some constant $C_4$.  Finally, the probability that some set of size between $1$ and $\alpha n$ is not $(\mathcal{T}+1)$-reachable is upper bounded by
\begin{equation*}
\sum_{m =1}^{\lfloor\alpha n\rfloor}\mathbb{P}_m \le \sum_{m =1}^{\lfloor\alpha n\rfloor}\left(C_4e^{-c}\right)^m \le \sum_{m =1}^{
\infty}\left(C_4e^{-c}\right)^m = \frac{C_4e^{-c}}{1-C_4e^{-c}}.
\end{equation*}
Since $c \rightarrow \infty$, we see that this probability goes to zero asymptotically, and thus the graph will be $(\mathcal{T}+1)$-robust asymptotically almost surely.
\end{proof}

Note that the filtering threshold distribution in the above theorem is allowed to have support that extends past the minimum degree $r$ of the network.  The key is to ensure that the probability that a node gets assigned a threshold above its degree goes to zero.  The above result characterizes a condition on the distribution (given by \eqref{eq:ER_filtering_threshold}) to ensure this. This result generalizes the result on ER graphs given in \cite{zhang2015notion}, which focused on the case of all nodes having the same threshold (below the minimum degree of the graph).  A byproduct of the above proof is that all subsets of vertices of size up to $\alpha n$ are guaranteed to be $(\mathcal{T}+1)$-reachable a.a.s., where $\alpha$ is such that $\ln\ln n = o((1-\alpha)\ln n)$.  For instance, $\alpha = 1 - \frac{1}{(\ln n)^{\epsilon}}$ for $\epsilon \in (0,1)$ satisfies this condition.

%% file: interdep.tex
In the previous section, we considered local-filtering based opinion dynamics in Erd\H os-R\'enyi graphs, where every possible edge appears independently with the same probability.  In this section, we extend our discussion to random graphs that consist of a set of communities (or subnetworks), where only the edges between communities appear independently with a certain probability $p$.  We will remain agnostic about the intra-community topologies.  

\begin{definition}
A {\bf random interdependent network} consists of $k$ subnetworks $\mathcal{G}_i = \{\mathcal{V}_i, \mathcal{E}_i\}$, $i = 1, \ldots, k$, along with a set of inter-network edges 
$$
\mathcal{E}_I \subseteq \bigcup_{i \ne j} \left\{\mathcal{V}_i \times \mathcal{V}_j\right\}.
$$ 
Specifically, $\mathcal{E}_I$ is obtained by placing each edge in $\bigcup_{i \ne j} \left\{\mathcal{V}_i \times \mathcal{V}_j\right\}$ independently with a probability $p$ (which can be a function of the  number of nodes in the network).
\end{definition}

In the rest of this section, we will assume that $|\mathcal{V}_i| = n$ for $i = 1, \ldots, k$, so that the random interdependent network has $nk$ nodes in total.  Note that an Erd\H os-R\'enyi graph on $nk$ nodes is a special case of the above definition, where each edge within each subnetwork is also placed with probability $p$. When the subnetwork edges are placed independently with a probability $\bar{p}$ that is different from $p$, such networks are known as {\it stochastic block models} \cite{decelle2011asymptotic,LeLarge2015}.  We will make no such assumptions on the intra-network topology here, however.  

Recent work has studied various structural properties of random interdependent networks, including algebraic connectivity and robustness \cite{shahrivar2015robustness}.  Here, as in the previous section, we consider the case where each node in the network has a personal opinion threshold, drawn from a distribution $q$ (i.e., $q(t)$ is the probability of a given node having threshold $t$).  We have the following result.

\begin{theorem}
Fix $k \in \mathbb{Z}_{\ge 2}$, and consider a random interdependent network on $nk$ nodes where each inter-network edge is placed independently with probability 
\begin{equation}
p(n) = \frac{\ln n + (r-1)\ln\ln n + c(n)}{(k-1)n},
\label{eq:RIN_edge_probability}
\end{equation}
where $c(n) \rightarrow \infty$ and $c(n) = o(\ln\ln n)$.  Suppose the opinion threshold distribution has support $[0, \bar{r}]$ for some $\bar{r}\in \mathbb{N}$ and  satisfies
$$
q(t) = O\left(\frac{1}{(\ln n)^{t-r +1}}\right)
$$
for $t \ge r$.  Then the random interdependent network with opinion thresholds drawn from $q(\cdot)$ will facilitate consensus under the local-filtering opinion dynamics a.a.s.
\end{theorem}

\begin{proof}
The proof follows in a similar manner to that for Erd\H os-R\'enyi graphs, although additional care must be taken to handle the fact that we are only placing edges between the subnetworks. 

To this end, consider a set $S \subset \cup_{i = 1}^{k}\mathcal{V}_i$ consisting of $m$ nodes, where $1 \le m \le \alpha nk$ (for some function $\alpha$ that goes to $1$).  Denote $|S \cap \mathcal{V}_i| = m_i$.  Consider some node $v \in \mathcal{S}$, and suppose $v \in \mathcal{V}_i$.  There are $n(k-1)- (m - m_i)$ nodes that are not in set $S$ or $\mathcal{V}_i$ for $v$ to connect to.  Thus, the probability that $v$ has fewer than $t+1$ neighbors outside $S \cup \mathcal{V}_i$ (where $t$ is drawn from the distribution $q(\cdot)$) is given by
\begin{equation*}
\sum_{t =0}^{\bar{r}}q(t)\sum_{j=0}^{t}{{n(k-1)- (m - m_i)}\choose{j}}p^j (1-p)^{n(k-1)- (m - m_i) - j}.
\end{equation*}
This is upper bounded by
$$
\sum_{t =0}^{\bar{r}}q(t)\sum_{j=0}^{t}{{n(k-1)}\choose{j}}p^j(1-p)^{n(k-1) - m - j}.
$$
Note that this is now in the same form as the probability of a given node having fewer than $t+1$ neighbors outside its set in ER graphs (given by \eqref{eq:single_node_ER}) with the only exception being that $n$ in \eqref{eq:single_node_ER} is replaced by $n(k-1)$.  Thus, the rest of the proof follows in the same manner as that proof (with the substitution of $n$ by $n(k-1)$), and thus the result for the random interdependent network follows by replacing $n$ by $n(k-1)$ in the probability $p(n)$ for ER graphs.  
\end{proof}

Comparing the edge probability for the random interdependent network in \eqref{eq:RIN_edge_probability} to that for the ER network in \eqref{eq:ER_edge_prob}, we see that the price paid for being agnostic about the subnetwork topology is an increase by a factor of $\frac{k}{k-1}$ in the edge formation probabilities in the former case (after scaling \eqref{eq:ER_edge_prob} to pertain to an ER graph on $nk$ nodes).  

%% file: hetero.tex
The graphs that we considered in the previous two sections were homogeneous, in the sense that each of the randomly chosen edges was placed with the same probability $p$ (although the random interdependent networks were allowed to have arbitrary topologies inside the subnetworks).  In this section, we discuss the extension of the results in the previous sections to random graphs with potentially different probabilities on each edge (e.g., as in the {\bf expected degree} random graph model \cite{chung2006complex}).

\begin{theorem}
Consider the undirected random graph where each edge $(v_i, v_j)$ is present independently with probability $p_{ij}(n)$.  Suppose that 
$$
p_{ij}(n) \ge \frac{\ln n + (r-1)\ln\ln n + c(n)}{n}
$$
for all $v_i, v_j \in \mathcal{V}$ (with $v_i \ne v_j$) and that $c(n) = o(\ln \ln n)$ with $c(n) \rightarrow \infty$.   Suppose the opinion threshold distribution has support $[0, \bar{r}]$ for some $\bar{r}\in \mathbb{N}$ and  satisfies
$$
q(t) = O\left(\frac{1}{(\ln n)^{t-r +1}}\right)
$$
for $t \ge r$. Then the resulting network facilitates consensus under the local-filtering opinion dynamics a.a.s.
\end{theorem}

\begin{proof}
The proof follows a standard coupling argument, relying on the monotonicity of the $(\mathcal{T}+1)$-robustness property (i.e., adding additional edges to a $(\mathcal{T}+1)$-robust graph maintains that property). 

Let $p(n)$ be as in \eqref{eq:ER_edge_prob}, and note that $p_{ij}(n) \ge p(n)$ for all $v_i, v_j \in \mathcal{V}$.  We create two networks $\mathcal{G}_1$ and $\mathcal{G}_2$ as follows.  For each edge $(v_i,v_j)$, we flip a coin that lands heads with probability $p(n)$.  If the coin lands heads, we place the edge $(v_i,v_j)$ in both $\mathcal{G}_1$ and $\mathcal{G}_2$.  If the coin lands tails, we do not place the edge in $\mathcal{G}_1$.  We then flip another coin that lands heads with probability $\frac{p_{ij}(n)-p(n)}{1-p(n)}$.  If this coin lands heads, we place the edge in graph $\mathcal{G}_2$, and do not place the edge otherwise.  We do this for all edges in the two graphs.

It is easy to see that graph $\mathcal{G}_1$ is an Erd\H os-R\'enyi graph with edge probability $p(n)$, and that each edge in graph $\mathcal{G}_2$ appears with probability $p_{ij}(n)$.  Furthermore, $\mathcal{G}_1$ is a subgraph of $\mathcal{G}_2$.  Draw the filtering thresholds for graph $\mathcal{G}_1$ from the distribution $q(t)$, and let the filtering thresholds for graph $\mathcal{G}_2$ be the same as the ones in $\mathcal{G}_1$.  By Theorem~\ref{thm:ER_graph_filtering_thresholds}, graph $\mathcal{G}_1$ will be $(\mathcal{T}+1)$-robust a.a.s., and thus graph $\mathcal{G}_2$ will be $(\mathcal{T}+1)$-robust a.a.s.  This concludes the proof.
\end{proof}